\documentclass[1p,final]{elsarticle}
\usepackage{amsfonts,color,morefloats,pslatex}
\usepackage{amssymb,amsthm, amsmath,latexsym}
\usepackage[pagewise]{lineno}
\usepackage{textcomp} 
\usepackage[latin1]{inputenc}

\newtheorem{theorem}{Theorem}
\newtheorem{lemma}[theorem]{Lemma}
\newtheorem{proposition}[theorem]{Proposition}

\newtheorem{remark}{Remark}

\begin{document}

\begin{frontmatter}



\title{Codebooks from generalized bent
$\mathbb{Z}_4$-valued quadratic forms
}


\author[Qi]{Yanfeng Qi}
\ead{qiyanfeng07@163.com}
\author[Mesnager]{Sihem Mesnager}
\ead{smesnager@univ-paris8.fr}
\author[Tang]{Chunming Tang
\footnote{Corresponding author: School of Mathematics and Information, China West Normal University;   Email: tangchunmingmath@163.com}}
\ead{tangchunmingmath@163.com}
\address[Qi]{School of Science, Hangzhou Dianzi University, Hangzhou, Zhejiang, 310018, China}
\address[Mesnager]{Department of Mathematics, University of Paris VIII, 93526 Saint-Denis, France, with LAGA UMR 7539,
CNRS, Sorbonne Paris Cit¨¦, University of Paris XIII, 93430 Paris, France, and also with Telecom ParisTech,
75013 Paris, France}
\address[Tang]{School of Mathematics and Information, China West Normal University, Nanchong, Sichuan,  637002, China, and
  Department of Mathematics, The Hong Kong University of Science and Technology, Clear Water Bay, Kowloon, Hong Kong, China}

\begin{abstract}
Codebooks with small inner-product
correlation have application in
unitary space-time modulations, multiple description coding over erasure channels,
direct spread code division multiple access communications,
compressed sensing, and coding theory. It is interesting to construct codebooks (asymptotically)
achieving the Welch bound or the Levenshtein
bound. This paper presented a class of generalized bent $\mathbb{Z}_4$-valued quadratic forms, which contain  functions of  Heng and Yue (Optimal codebooks achieving the Levenshtein bound from generalized bent functions over $\mathbb{Z}_4$. Cryptogr. Commun. 9(1), 41-53, 2017). By using these
generalized bent $\mathbb{Z}_4$-valued quadratic forms, we constructs optimal
codebooks achieving the Levenshtein bound. These codebooks have parameters  $(2^{2m}+2^m,2^m)$ and  alphabet size $6$.
\end{abstract}


\end{frontmatter}

\section{Introduction}
Applied in unitary space-time modulations, multiple description coding over erasure channels,
direct spread code division multiple access communications,
compressed sensing, and coding theory \cite{CCKS97,MM93}, an $(N,K)$ codebook  is a
signal set
$(\mathbf{c}_0,\ldots,\mathbf{c}_{N-1})$, where
$\mathbf{c}_0\ldots,\mathbf{c}_{N-1}$ are unit norm $1\times K$ complex vectors over an alphabet. As a performance measure of a codebook in practical applications, the maximum crosscorrelation amplitude of an $(N, K)$ codebook $\mathcal{C}$ is
defined by
$$
I_{max}(\mathcal{C})=max_{0\leq i<j\leq N-1}
(\mathbf{c}_i\mathbf{c_j}^H),
$$
where $\mathbf{c}_j^H$ is the conjugate transpose of $\mathbf{c}_j$.
There are two well known bounds of codebooks:
\begin{itemize}
\item the Welch bound \cite{W74}:
$I_{max}(\mathcal{C})\geq \sqrt{\frac{N-K}{(N-1)K}}$, where $N\geq K$;
\item the  Levenstein bounds
\cite{KL78,L83}:
\begin{itemize}
\item $I_{\max }(\mathcal{C}) \geq \sqrt{\frac{3 N-K^{2}-2 K}{(N-K)(K+2)}}$ for any real-valued codebook, where $N>K(K+1)/2$;
\item $I_{\max }(\mathcal{C}) \geq \sqrt{\frac{2 N-K^{2}-K}{(N-K)(K+1)}}$
    for any complex-valued codebook, where
    $N>K^2$.
\end{itemize}
\end{itemize}
A codebook achieving the Welch bound is also
called a maximum-Welch-bound-equality (MWBE) codebook, which is known as an equiangular tight frame \cite{C03}. The construction of MWBE codebooks is equivalent
to line packing in Grassmannian spaces
\cite{SH03}.
It is difficult to construct  MWBE codebooks
\cite{S99}.
Known $(N,K)$ MWBE codebooks are listed:
\begin{itemize}
\item $(N,N)$ orthogonal MWBE codebooks
\cite{S99,XZG05}, where $N>1$;
\item $(N,N-1)$ MWBE codebooks
from discrete Fourier transformation matrices
\cite{S99,XZG05} or m-sequences \cite{S99}, where $N>1$;
\item $(N, K)$ MWBE codebooks from conference matrices \cite{CHS96,SH03}, where $N=2K$, $N=2^{d+1}$ or
    $N=p^d+1$ for a prime $p$;
\item $(N, K)$ MWBE codebooks from
difference sets in cyclic groups \cite{XZG05} and abelian group \cite{D06,DF07}, or Steiner systems \cite{FMT12}.
\end{itemize}
Almost optimal codebooks asymptotically achieving the Welch bound are constructed in \cite{DF08,HY13,LYH15,ZF09,ZF12}.
When $N>K(K+1)/2$ (resp. $N>K^2$), there are no
real-valued (complex-valued) MWBE codebooks. The  Levenshtein bounds hold for
$N>K(K+1)/2$ or $N>K^2$. Some known
$(N,K)$ codebooks achieving the Levenstein Bound are listed:
\begin{itemize}
\item $(2^{2m-1}+2^m,2^m)$ codebooks with alphabet $4$  \cite{CCKS97,WF89,XDM15,ZDL14},  where $m$ is even;
\item $(p^{2m}+p^m, p^m)$ codebooks with alphabet size $p+2$ \cite{DY07,WF89,XDM15}, where $p$ is odd;
\item  $(2^{2m}+2^m,2^m)$ codebooks with alphabet $6$ \cite{CCKS97,HY17}.
\end{itemize}
These codebooks can  constructed from
binary Kerdock codes \cite{CCKS97,WF89,XDM15}, perfect nonlinear functions \cite{DY07,WF89},  bent functions \cite{ZDL14},
$\mathbb{Z}_4$-Kerdock codes \cite{CCKS97}, and $\mathbb{Z}_4$-valued quadratic forms \cite{HY17}.  It is interesting
to construct optimal codebooks achieving the Levenstein Bound with different
parameters from different methods.
Zhou et al. \cite{ZDL14} presented codebooks
achieving the Levenstein bound from bent functions of the form:
$$
f_{\gamma}\left(x_{1}, x_{2}\right)=\sum_{i=1}^{\frac{m-2}{2}} \operatorname{tr}_{1}^{m-1}\left(x_{1}^{2^{i}+1}\right)+\sum_{i=1}^{\frac{\ell-1}{2}} \operatorname{tr}_{1}^{m-1}\left(\left(\gamma x_{1}\right)^{2^{e i}+1}\right)+x_{2} \operatorname{tr}_{1}^{m-1}\left(x_{1}\right)
$$
where $x_1\in \mathbb{F}_{2^{m-1}}$,
$x_2\in \mathbb{F}_2$, and $\mathrm{tr}_1^{m-1}$ is the trace function from $\mathbb{F}_{2^{m-1}}$ to $\mathbb{F}_2$.
Heng and Yue \cite{HY17} generalized their results to codebooks from generalized bent  $\mathbb{Z}_4$-valued quadratic forms.
Ding et al. \cite{DMTX18} presented the notation of cyclic bent functions, gave a class of cyclic bent functions containing bent functions in
\cite{ZDL14}, and used cyclic bent functions in the construction of good mutually unbiased bases (MUBs), codebooks and sequence families.

Motivated by their methods, this paper generalizes the construction of Heng and Yue
\cite{HY17} and constructs optimal codebooks from a class of generalized bent $\mathbb{Z}_4$-valued quadratic forms.

The rest of the paper is organized as follows.
Section 2 introduces some basic results on
Galois rings, $\mathbb{Z}_4$-valued quadratic forms and codebooks from generalized bent functions. Section 3 presents the construction of optimal codebooks from a class of generalized bent $\mathbb{Z}_4$-valued quadratic forms.  Section 4 makes a conclusion.

\section{Preliminaries}
In this section, we introduce some results on
Galois rings, $\mathbb{Z}_4$-valued
quadratic forms, and codebooks from generalized bent $\mathbb{Z}_4$-quadratic forms.

\subsection{Galois rings}
Let $m$ be a positive integer. Let
$R=\mathbb{GR}(4,m)$ be the Galois ring $\mathbb{Z}_4[x]/(f)$, where
$\mathbb{Z}_4$ is the   ring of integers modulo $4$ and
$f$ is a monic basic irreducible
polynomial of degree $m$ in
$\mathbb{Z}_4[x]$. More results on
Galois rings can be found in
\cite{HK98,HY15,M74,S09}.

Define $F=\{z\in R: z^{2^m}=z\}$, which is called the set of Teichmuller representatives in $R$. Then  any
$z\in R$ can be uniquely represented by  $z=a+2b$, where $a,b\in F$.
For two elements $a,b\in F$,
Define $a\oplus b=(a+b)^{2^m}
=a+b+2\sqrt{ab}$.
Then $(F,\oplus,\cdot)$ is a Galois field of
size $2^m$.

Let $
\mu: \mathbb{Z}_4\longrightarrow
\mathbb{Z}_2$ denote the
modulo-$2$ reduction. This mapping induces the following mapping from $R$ to $\mathbb{F}_{2^m}$
\begin{eqnarray*}
\mu: R&\longrightarrow& \mathbb{F}_{2^m} \\
z&\longmapsto& \mu(z)=\overline{z}
\end{eqnarray*}
Further, $\mu: F\longrightarrow
\mathbb{F}_{2^m}$ is an isomorphism.

The Frobenius automorphism $\sigma$ on
$F$ is given by $\sigma(x)=x^2$. The trace function from $R$ to $\mathbb{Z}_4$ is defined by
\begin{eqnarray*}
\mathrm{Tr}_1^m: R&\longrightarrow&
\mathbb{Z}_4\\
z=a+2b&\longmapsto& \sum_{i=0}^{m-1}
(\sigma^i(a)+2\sigma^i(b))
\end{eqnarray*}
Then $\mathrm{Tr}_1^m(\alpha z_1+
\beta z_2)=\alpha \mathrm{Tr}_1^m(z_1)
+\beta \mathrm{Tr}_1^m(z_2)$, where
$z_1,z_2\in R$ and $\alpha,\beta\in
\mathbb{Z}_4$. Let $\mathrm{tr}_1^m$ be the trace function from $\mathbb{F}_{2^m}$ to
$\mathbb{F}_2$.  For any $x\in F$,  $2\mathrm{Tr}_1^m(x)
=2\mathrm{tr}_1^m(x)$.
The follow lemma \cite{HKCSS94} shows that
$\overline{x_1}\neq \overline{x_2}$
for two different elements $x_1,x_2\in F$.
\begin{lemma}
Let $F^*$ be the set of nonzero elements
of $F$, which is a multiplicative group.
Let  $F^*=\langle \xi  \rangle$.
For $0\leq i<j< 2^m-1$,
$\xi^i\pm \xi^j$ is invertible.
\end{lemma}

\subsection{$\mathbb{Z}_4$-valued
quadratic forms}

Let
$K=\{z\in \mathbb{Z}_4: z=z^2\}$.
Then any $z\in \mathbb{Z}_4$ can be uniquely
represented by $z=x+2y$, where $x,y \in K$.
An operation $\oplus$ on $K$ can be denoted by
$x\oplus y=(x+y)^2$.
Then, $(K,\oplus,+)$ is the finite field of size $2$.

A symmetric bilinear form on $F$ is a mapping
$B: F\times F\longrightarrow K$ satisfying
\begin{itemize}
\item the symmetry condition:
$B(x,y)=B(y,x)$;
\item  the bilinearity condition:
$B(\alpha x_1\oplus \beta x_2, y)
=\alpha B(x_1,y)+ \beta B(x_2,y)$, where
$\alpha, \beta\in K$.
\end{itemize}
A symmetric bilinear for $B$ is alternating if
$B(x,x)=0$ for all $x\in F$. Otherwise, it is
nonalternating. The radical $rad(B)$ is the set
$$
rad(B)=\{x\in F: B(x,y)=0 ~\text{for all}~y \in F\}.
$$
The radical is a vector space over $K$. The rank of $B$ is defined by
$$
rank(B)=m-dim_K(rad(B)).
$$

A $\mathbb{Z}_4$-valued quadratic form on $F$
\cite{B72} is a mapping  $Q: F\longrightarrow
\mathbb{Z}_4$ satisfying
\begin{itemize}
\item $Q(\alpha x)=\alpha^2 Q(x)$ for $\alpha\in K$;
\item $Q(x\oplus y)=Q(x)+Q(y)+2B(x,y)$, where
$B$ is a symmetric bilinear form.
\end{itemize}
A $\mathbb{Z}_4$-quadratic form $Q$ is alternating if
its associated bilinear form $B$ is alternating. Otherwise, $Q$ is nonalternating. The rank
of $Q$ is the rank of its associated bilinear form. Note that $Q(0)=0$ and
$2Q(x)=2B(x,x)$. The Walsh transform of $Q$ is defined by
$$
\chi_{Q}(\lambda)=\sum_{x \in F}(\sqrt{-1})^{Q(x)+2 \operatorname{Tr}_{1}^{m}(\lambda x)}
$$
where $\lambda\in F$.
The multiset $\{\chi_{Q}(\lambda):
\lambda \in F\}$ depends only on the rank of
$Q$. For an alternating $Q$, the following theorem \cite{HK98} gives
the distribution of values in the multiset
$\{\chi_{Q}(\lambda):
\lambda \in F\}$.
\begin{theorem}\label{thm-x}
Let $Q: F\longrightarrow \mathbb{Z}_4$
be an alternating $\mathbb{Z}_4$-valued quadratic form of rank $r$. The distribution of values in the multiset $\{\chi_{Q}(\lambda):
\lambda \in F\}$ is given in Tabel \ref{t1}.
\begin{table}[htbp]
\centering
\caption{The distribution of values in the multiset $\{\chi_{Q}(\lambda):
\lambda \in F\}$}\label{t1}
\begin{tabular}{|c|c|}
  \hline
  value & frequency \\
\hline
  $0$ & $2^m-2^r$ \\
\hline
  $\pm 2^{m-r/2}$ & $2^{r-1}\pm
2^{r/2-1}$  \\
  \hline
\end{tabular}
\end{table}
\end{theorem}
For a  nonalternating $Q$, the distribution of values in the multiset $\{\chi_{Q}(\lambda):
\lambda \in F\}$ is given in Theorem 5
\cite{S09}.

A $\mathbb{Z}_4$-quadratic form $Q$ is
generalized bent if $\left|\chi_{Q}(\lambda)\right|=2^{m/2}$
for all $\lambda \in F$
\cite{CM16,M15}.  Let $\mathcal{GF}_m$ be
the set of all generalized bent
$\mathbb{Z}_4$-quadratic forms over $F$.
The following lemma \cite{LTH14}
gives a characterization of generalized bent
$\mathbb{Z}_4$-quadratic forms.
\begin{lemma}\label{gbent}
A $\mathbb{Z}_4$-quadratic form $Q(x)$ has full rank $m$ if and only if it is generalized bent.
\end{lemma}

\subsection{Codebooks from generalized bent
$\mathbb{Z}_4$-quadratic forms}

Let $QF$ be an $n$-set of generalize bent
$\mathbb{Z}_4$-quadratic forms from
$F$ to $\mathbb{Z}_4$ such that
the difference of arbitrary two
distinct quadratic forms in $QF$
 is generalized bent. Let $E_n=
\{\mathbf{e}_i: 1\leq n\}$ be the standard
basis of the $2^m$-dimensional Hilbert space,
where $
\mathbf{e}_i=(0,\ldots,0,1,0,\ldots,0  )
$ is a vector with only the $i$-th entry being nonzero.
  Construct the following
codebook from $QF$:
\begin{align}\label{codebook}
\mathcal{C}_{QF}
=\bigcup_{Q\in QF}S_{Q}\bigcup S_0
\bigcup E_{2^m},
\end{align}
where
$$
S_Q=\left\{\frac{1}{\sqrt{2^m}}
\left((\sqrt{-1})^{Q(\xi)+2\mathrm{Tr}_1^m(
\lambda \xi)} \right)_{\xi \in F}:
\lambda \in F \right\}
$$
and
$$
S_0=\left\{\frac{1}{\sqrt{2^m}}
\left((\sqrt{-1})^{2\mathrm{Tr}_1^m(
\lambda \xi)} \right)_{\xi \in F}:
\lambda \in F \right\}.
$$
The following theorem
\cite{HY17} gives parameters of the codebook constructed from $QF$.
\begin{theorem}\label{optimal}
Let $QF$ be an $n$-set of generalize bent
$\mathbb{Z}_4$-quadratic forms from
$F$ to $\mathbb{Z}_4$ such that
the difference of arbitrary two
distinct quadratic forms in $QF$
 is generalized bent. Let
$\mathcal{C}_{QF}$ be the codebook
constructed in (\ref{codebook}).
Then  $\mathcal{C}_{QF}$ is a
$((n+1)2^m+2^m,2^m)$ codebook with
$I_{max}(\mathcal{C}_{QF})
=\frac{1}{\sqrt{2^m}}$ and alphabet size $6$.
Further, the codebook  $\mathcal{C}_{QF}$
is an optimal codebook achieving
 the Levenshtein bound
if and only if $n=2^m-1$.
\end{theorem}
From Theorem \ref{optimal}, in order to
construct optimal codebooks, we just need to
give the set $QF$ of size $2^m-1$. Some known
results on the set $QF=\{Q_a: a\in F^*\}$
\cite{HY17} have
been given below:
\begin{itemize}
\item $Q_{a}(x)=\operatorname{Tr}_{1}^{m}
(a(1+2 \eta) x)$;
\item $Q_{a}(x)=\operatorname{Tr}_{1}^{m}(a x)+
2 P(\gamma a x)$, where
$P(x)= {\sum_{i=1}^{\frac{s-1}{2}}
\operatorname{Tr}_{1}^{m}
\left(x^{2^{ek i}+1}\right)}$ and $s$ is odd;
\item $Q_{a}(x)=\operatorname{Tr}_{1}^{m}(a(1+2 \eta) x)+
2 P(\gamma a x)$, where
$P(x)= {\sum_{i=1}^{\frac{s-1}{2}}
\operatorname{Tr}_{1}^{m}
\left(x^{2^{ek i}+1}\right)}$ and $s$ is odd.
\end{itemize}

\section{Codebooks from a class of generalize bent
$\mathbb{Z}_4$-quadratic forms}

In this section, we present a class of
generalize bent
$\mathbb{Z}_4$-quadratic forms and construct codebooks from these quadratic forms.

Let $l$ be an positive integer.
Let $e_0\leq e_1<\cdots < e_l$, where
$e_0=1$, $e_i\mid e_{i+1}$ for $i
\in \{0,\ldots,l-1\}$,  and $e_l=m$. Let $f_i=\frac{m}{e_i}$ for
$i\in \{1,\ldots, l-1\}$, where $f_i$ is odd.
For $j\in \{1,\ldots,l-1\}$, define
$$
Q_j(x)=\mathrm{Tr}_1^m(\sum_{i=1}^{
\frac{f_j-1}{2}}x^{2^{ie_j}+1}).
$$
Take
$\gamma_0,\gamma_1,\ldots,
\gamma_{l-1}\in F$, where
$\overline{\gamma_0}=1$,
$\overline{\gamma_j}\in \mathbb{F}_{2^{e_j}}$,
and
$1+\sum_{j=1}^{t}\overline{\gamma_j}^2\neq 0$
for $t\leq l-1$.
For any $a\in F^*$, define
\begin{equation}\label{f_a}
f_a(x)=\mathrm{Tr}_1^m(ax)+
2\sum_{j=1}^{l-1}Q_j(\gamma_j ax).
\end{equation}
 For $j\in \{1,\ldots,l-1\}$, define
the Boolean function  over
$\mathbb{F}_{2^m}$
$$
q_j(x)=\mathrm{tr}_1^m(\sum_{i=1}^{
\frac{f_j-1}{2}}x^{2^{ie_j}+1}).
$$
It is a  quadratic form and its symplectic form is
$$
B_{q_j}(x,y)=q_j(x)+q_j(y)+q_j(x+y),
$$
where $x,y\in \mathbb{F}_{2^m}$.
Further, $B_{q_j}(x,y)=
\mathrm{tr}_1^m\left(y
\left(\mathrm{tr}_{e_j}^m(x)+x \right) \right)
$ \cite{KN03}.
We will use these functions defined in
(\ref{f_a}) to construct
codebooks. Some lemmas are given first.

\begin{lemma}\label{Bfa}
Let $f_a$ be defined in (\ref{f_a}).
Then $2B_{f_a}(x,y)=
2\mathrm{tr}_1^m
 \left(\overline{y}
 \left(\overline{a}^2\overline{x}+
 \sum_{j=1}^{l-1}\overline{
 \gamma_ja}(\mathrm{tr}_{e_j}^m(\overline{\gamma_j ax})+\overline{\gamma_j ax})  \right)
\right)$.
\end{lemma}
\begin{proof}
We have
\begin{align*}
2B_{f_a}(x,y)=&f_a(x\oplus y)
-f_a(x)-f_a(y)\\
=& \mathrm{Tr}_1^m(ax\oplus ay)
-\mathrm{Tr}_1^m(ax)-\mathrm{Tr}_1^m(ay)
+2\sum_{j=1}^{l-1}
\left(Q_j(\gamma_j ax\oplus \gamma_j ay)
-Q_j(\gamma_j ax)-Q_j(\gamma_j ay)\right)\\
=& 2\mathrm{Tr}_1^m(a\sqrt{xy})
+\sum_{j=1}^{l-1}
2\left(Q_j(\gamma_j ax\oplus \gamma_j ay)
-Q_j(\gamma_j ax)-Q_j(\gamma_j ay)\right)\\
=& 2\mathrm{tr}_1^m(\overline{a}^2\overline{xy})
+\sum_{j=1}^{l-1}
2\left(q_j(\overline{\gamma_j ax}+
\overline{\gamma_j ay})
+q_j(\overline{\gamma_j ax})+q_j(
\overline{\gamma_j ay})\right)\\
=& 2\mathrm{tr}_1^m(\overline{a}^2\overline{xy})
+\sum_{j=1}^{l-1}
2B_{q_j}(\overline{\gamma_j ax},
\overline{\gamma_j ay})\\
=& 2\mathrm{tr}_1^m(\overline{a}^2\overline{xy})
+\sum_{j=1}^{l-1}
2\mathrm{tr}_1^m\left(\overline{\gamma_j ay}
(\mathrm{tr}_{e_j}^m(\overline{\gamma_j ax})+\overline{\gamma_j ax})
 \right)\\
 =& 2\mathrm{tr}_1^m
 \left(\overline{y}
 \left(\overline{a}^2\overline{x}+
 \sum_{j=1}^{l-1}\overline{
 \gamma_ja}(\mathrm{tr}_{e_j}^m(\overline{\gamma_j ax})+\overline{\gamma_j ax})  \right)  \right).
\end{align*}
Hence, this lemma follows.
\end{proof}


\begin{lemma}\label{Bfab}
Let $a,b\in F$ and $f_a,f_b$ be defined in
(\ref{f_a}), where
$a\neq b$. Then
$$
2B_{f_a-f_b}(x,y)=2\mathrm{tr}_1^m
 \left(\overline{y}
 \left({(\overline{a^2}+\overline{b^2})
 \overline{x}}+
 \sum_{j=1}^{l-1}\left(\overline{
 \gamma_ja}(\mathrm{tr}_{e_j}^m(\overline{\gamma_j ax})+\overline{\gamma_j ax})+
 \overline{
 \gamma_jb}(\mathrm{tr}_{e_j}^m(\overline{\gamma_j bx})+\overline{\gamma_j bx})\right)  \right)  \right).$$
\end{lemma}
\begin{proof}
Note that $2B_{f_a-f_b}(x,y)=2B_{f_a}(x,y)-2B_{f_b}(x,y)$. From Lemma \ref{Bfa}, this lemma follows.
\end{proof}

\begin{lemma}\label{radB}
Let $a,b\in F$ and $f_a,f_b$ be defined in
(\ref{f_a}), where
$a\neq b$. Then $rad(B_{f_a-f_b})=\{0\}$.
\end{lemma}
\begin{proof}
Let $x\in rad(B_{f_a-f_b})$, we just need to prove that $x=0$.

Suppose that $x\in rad(B_{f_a-f_b})$ and
$x\neq 0$. By Lemma \ref{Bfab}, $\overline{x}$ is a nonzero solution of
\begin{align*}
(\overline{a^2}+\overline{b^2})
\overline{x}+
 \sum_{j=1}^{l-1}\left(\overline{
 \gamma_ja}(\mathrm{tr}_{e_j}^m(\overline{\gamma_j ax})+\overline{\gamma_j ax})+
 \overline{
 \gamma_jb}(\mathrm{tr}_{e_j}^m(\overline{\gamma_j bx})+\overline{\gamma_j bx})\right)
=0.
\end{align*}
We have
$$
(1+\sum_{j=1}^{l-1}\overline{\gamma_j}^2)
(\overline{a^2}+\overline{b^2})
\overline{x}+
\sum_{j=1}^{l-1} \overline{\gamma_j^2}
\left(\overline{a}\mathrm{tr}_{e_j}^m(
\overline{ax})+\overline{b}\mathrm{tr}_{e_j}^m(
\overline{bx}) \right)=0
$$
and
$$
(1+\sum_{j=1}^{l-1}\overline{\gamma_j}^2)
(\overline{a^2}+\overline{b^2})
\overline{x}^2+
\sum_{j=1}^{l-1} \overline{\gamma_j^2} \overline{x}
\left(\overline{a}\mathrm{tr}_{e_j}^m(
\overline{ax})+\overline{b}\mathrm{tr}_{e_j}^m(
\overline{bx}) \right)=0.
$$

Note that $\mathrm{tr}_{e_j}^m(y^2)=
(\mathrm{tr}_{e_j}^m(y))^2$ for any
$y\in \mathbb{F}_{2^m}$.
Let $$A=(1+\sum_{j=1}^{l-1}\overline{\gamma_j}^2)
(\overline{a^2}+\overline{b^2})
\overline{x}^2+
\sum_{j=1}^{l-1} \overline{\gamma_j^2} \overline{x}
\left(\overline{a}\mathrm{tr}_{e_j}^m(
\overline{ax})+\overline{b}\mathrm{tr}_{e_j}^m(
\overline{bx}) \right).$$
Then we have
\begin{align*}
\mathrm{tr}_{e_{l-1}}^m(A)
=&
(1+\sum_{j=1}^{l-1}\overline{\gamma_j}^2)
\mathrm{tr}_{e_{l-1}}^m
\left
((\overline{a^2}+\overline{b^2})
\overline{x}^2\right) +
\sum_{j=1}^{l-2} \overline{\gamma_j^2}
\left(
\mathrm{tr}_{e_{l-1}}^m(\overline{a}\overline{x})
\mathrm{tr}_{e_j}^m(
\overline{ax})+\mathrm{tr}_{e_{l-1}}^m
(\overline{b}\overline{x})\mathrm{tr}_{e_j}^m(
\overline{bx}) \right)\\&
+ \overline{\gamma_{l-1}^2}\mathrm{tr}_{e_{l-1}}^m
\left
((\overline{a^2}+\overline{b^2})
\overline{x}^2\right)
\\
=& (1+\sum_{j=1}^{l-2}\overline{\gamma_j}^2)
\mathrm{tr}_{e_{l-1}}^m
\left
((\overline{a^2}+\overline{b^2})
\overline{x}^2\right) +
\sum_{j=1}^{l-2} \overline{\gamma_j^2}
\left(
\mathrm{tr}_{e_{l-1}}^m(\overline{a}\overline{x})
\mathrm{tr}_{e_j}^m(
\overline{ax})+\mathrm{tr}_{e_{l-1}}^m
(\overline{b}\overline{x})\mathrm{tr}_{e_j}^m(
\overline{bx}) \right).
\end{align*}
Since $f_{j+1}$  is odd, then $\mathrm{tr}_{e_j}^m
(\mathrm{tr}_{e_{j+1}}^m(y))
=\mathrm{tr}_{e_{j}}^m(y)$
for any
$y\in \mathbb{F}_{2^m}$.  We have
$$
\mathrm{tr}_{e_{1}}^m\left(
 \mathrm{tr}_{e_{2}}^m\left(
\cdots \mathrm{tr}_{e_{l-1}}^m\left(A
 \right) \cdots
 \right)\right)
= \mathrm{tr}_{e_1}^{m}((\overline{a}^2+
\overline{b}^2)\overline{x}^2)=
\left(
\mathrm{tr}_{e_1}^{m}((\overline{a}+
\overline{b})\overline{x})\right)^2=0.
$$
Hence, $\mathrm{tr}_{e_1}^{m}((\overline{a}+
\overline{b})\overline{x})=0$.
There exists   a $t$  such that
$
\left\{
  \begin{array}{l}
     \mathrm{tr}_{e_{t+1}}^{m}((\overline{a}+
\overline{b})\overline{x})\neq 0,\\
     \mathrm{tr}_{e_j}^{m}((\overline{a}+
\overline{b})\overline{x})=0,~
\text{for any}~1\leq j \leq t \,.
  \end{array}
\right.
$
For any $j\geq t+1$, we have
\begin{align*}
\mathrm{tr}_{e_{t+1}}^{m}\left(\overline{
 \gamma_jax}(\mathrm{tr}_{e_j}^m(\overline{\gamma_j ax})+\overline{\gamma_j ax})\right)
=& \mathrm{tr}_{e_{t+1}}^{e_j}\left(
\mathrm{tr}_{e_{j}}^{m}\left(
\overline{
 \gamma_jax}(\mathrm{tr}_{e_j}^m(\overline{\gamma_j ax})+\overline{\gamma_j ax})
\right)
\right)\\
=&
\mathrm{tr}_{e_{t+1}}^{e_j}
\left(
\left(
\mathrm{tr}_{e_{j}}^{m}(
\overline{\gamma_jax})\right)^2+
\mathrm{tr}_{e_{j}}^{m}\left(
(\overline{\gamma_jax})^2\right)
\right)
\\
=&0.
\end{align*}
Since $\mathrm{tr}_{e_j}^{m}((\overline{a}+
\overline{b})\overline{x})=0$
 for any $1\leq j \leq t$,
let $u_j=\mathrm{tr}_{e_j}^{m}(
\overline{a}\overline{x})
=\mathrm{tr}_{e_j}^{m}(
\overline{b}\overline{x})$, where
$1\leq j\leq t$.
Then
\begin{align*}
\mathrm{tr}_{e_{t+1}}^{m}(A)=&
\mathrm{tr}_{e_{t+1}}^{m}\left(
(\overline{a^2}+\overline{b^2})
\overline{x}^2\right)+
 \sum_{j=1}^{t}\mathrm{tr}_{e_{t+1}}^{m}\left(\overline{
 \gamma_jax}(\mathrm{tr}_{e_j}^m(\overline{\gamma_j ax})+\overline{\gamma_j ax})+
 \overline{
 \gamma_jbx}(\mathrm{tr}_{e_j}^m(\overline{\gamma_j bx})+\overline{\gamma_j bx})\right)\\
=&\left( \mathrm{tr}_{e_{t+1}}^{m}((\overline{a}+
\overline{b})\overline{x})\right)^2+
\mathrm{tr}_{e_{t+1}}^{m}((\overline{a}+
\overline{b})\overline{x})
\sum_{j=1}^{t}\overline{\gamma_j}^2
\left(
\mathrm{tr}_{e_{t+1}}^{m}((\overline{a}+
\overline{b})\overline{x})
+ u_j
\right)
\\
=&
\mathrm{tr}_{e_{t+1}}^{m}((\overline{a}+
\overline{b})\overline{x})
\left(
(1+\sum_{j=1}^{t}\overline{\gamma_j}^2)
\mathrm{tr}_{e_{t+1}}^{m}((\overline{a}+
\overline{b})\overline{x})
+\sum_{j=1}^{t}\overline{\gamma_j}^2u_j
\right)\\
=&0.
\end{align*}
Note that $\mathrm{tr}_{e_{t+1}}^{m}((\overline{a}+
\overline{b})\overline{x})\neq 0$ and
$1+\sum_{j=1}^{t}\overline{\gamma_j}^2\neq 0$.
We have $$(1+\sum_{j=1}^{t}\overline{\gamma_j}^2)
\mathrm{tr}_{e_{t+1}}^{m}((\overline{a}+
\overline{b})\overline{x})
+\sum_{j=1}^{t}\overline{\gamma_j}^2u_j=0$$
and
$$
\mathrm{tr}_{e_{t+1}}^{m}((\overline{a}+
\overline{b})\overline{x})=
\frac{\sum_{j=1}^{t}\overline{\gamma_j}^2u_j}{
1+\sum_{j=1}^{t}\overline{\gamma_j}^2}
\in \mathbb{F}_{2^{e_t}}.
$$
Note that $f_j$ is odd. We have
$$
0=\mathrm{tr}_{e_{t}}^{m}((\overline{a}+
\overline{b})\overline{x})
=\mathrm{tr}_{e_{t}}^{e_{t+1}}\left(
\mathrm{tr}_{e_{t+1}}^{m}
((\overline{a}+
\overline{b})\overline{x})
\right)
=\mathrm{tr}_{e_{t+1}}^{m}
((\overline{a}+
\overline{b})\overline{x})
\mathrm{tr}_{e_{t}}^{e_{t+1}}(1)
=\mathrm{tr}_{e_{t+1}}^{m}
((\overline{a}+
\overline{b})\overline{x}),
$$
which makes a contradiction with the definition of $t$.   Hence, $x=0$ if $x\in rad(B_{f_a-f_b})$. This lemma follows.
\end{proof}

\begin{theorem}\label{fgbent}
The set $S=\{f_a: a\in F^*\}$ is an
$n$-set of generalize bent
$\mathbb{Z}_4$-quadratic forms from
$F$ to $\mathbb{Z}_4$ such that
the difference of arbitrary two
distinct quadratic forms in $S$
 is generalized bent, where
$f_a$ is defined in (\ref{f_a}) and
$n=2^m-1$.
\end{theorem}
\begin{proof}
We just need to prove $f_a-f_b$ is
generalized bent for
two distinct $a$ and $b$ in $S$.
From Lemma \ref{gbent}, $f_a-f_b$ is generalized
bent if and only if $rank(f_a-f_b)=m$.
From Lemma \ref{radB}, we have
$rank(f_a-f_b)=rank(B_{f_a-f_b})=m-
dim(rad(B_{f_a-f_b}))=m$. Hence, this
theorem follows.
\end{proof}
We construct optimal codebooks from the set
$S=\{f_a: a\in F^*\}$ in the following theorem.
\begin{theorem}
Let $S=\{f_a: a\in F^*\}$, where $f_a$
is defined in (\ref{f_a}). Then the codebook
$\mathcal{C}_S$ constructed in
(\ref{codebook})  is a
$(2^{2m}+2^m,2^m)$ codebook with
$I_{max}(\mathcal{C}_{S})
=\frac{1}{\sqrt{2^m}}$ and alphabet size $6$,
which is an optimal codebook achieving
 the Levenshtein bound.
\end{theorem}
\begin{proof}
From Theorem \ref{optimal} and Theorem
\ref{fgbent}, this theorem follows.
\end{proof}
\begin{remark}
Let $\eta\in F$. Then for any $a\in F^*$, we have the following $\mathbb{Z}_4$-valued quadratic form:
\begin{align}\label{f'_a}
f'_a(x)=\mathrm{Tr}_1^m(a(1+2\eta)x)+
2\sum_{j=1}^{l-1}Q_j(\gamma_j ax).
\end{align}
From similar proof, we have that
the set $S'=\{f'_a(x): a\in F^*\}$ also satisfies
properties in Theorem \ref{fgbent}. Then
$\mathcal{C}_{S'}$ constructed in
(\ref{codebook})  is a
$(2^{2m}+2^m,2^m)$ codebook with
$I_{max}(\mathcal{C}_{S'})
=\frac{1}{\sqrt{2^m}}$ and alphabet size $6$,
which is an optimal codebook achieving
 the Levenshtein bound. When $l=2$,
 $e_1=t$ and $f_1=s$, we have the functions
   $\mathrm{Tr}_1^m(a(1+2\eta)x)+
2 Q_1(\gamma_1 ax)$, which are functions in \cite{HY17}.
These functions are used to construct
optimal families  of quadriphase sequences
\cite{JHTZ09}.
\end{remark}

We will give a connection between generalized bent $f'_a(x)$ and Boolean bent functions.
A $\mathbb{Z}_4$-valued form
$Q(x)$ has the representation
$$
Q(x)=q_1(\overline{x})+2q_2(\overline{x}),
$$
where $q_1$ and $q_2$ are two Boolean functions. We have the following Gray map
$\phi(Q)=q_1(u)v+q_2(u)$, which is a
Boolean function defined over
$\mathbb{F}_{2^m}\times \mathbb{F}_2$.
Then we have the following proposition.
\begin{proposition}
Let $m$ be odd and $f_a'(x)$ be a $\mathbb{Z}_4$-valued
quadratic form defined in (\ref{f'_a}). Then
$\phi(f_a')$ is a Boolean bent function over
$\mathbb{F}_{2^m}\times \mathbb{F}_2$.
\end{proposition}
\begin{proof}
Note that $\mathrm{Tr}_1^m(x)=
\mathrm{tr}_1^m(\overline{x})+
2q(\overline{x})$ \cite{HKCSS94}, where
$q(\overline{x})=
\sum_{i=1}^{\frac{m-1}{2}}\mathrm{tr}_1^m(
x^{2^i+1})$. Hence,
$f_a'(x)$ has the representation
$f_a'(x)=f_1(\overline{x})+2f_2(\overline{x})$, where $f_1$ and $f_2$ are two Boolean functions. In \cite{ST09}, $Q(x)=q_1(\overline{x})+2q_2(\overline{x})$ is generalized bent if and only if
$\phi(Q)$ is bent. Hence, this proposition follows.
\end{proof}

\section{Conclusion}

We present an $n$-set of generalized
bent $\mathbb{Z}_4$-valued quadratic forms such that any difference of two different elements in this set is also generalized bent, where
$n=2^m-1$. Using this $n$-set, we construct optimal codebooks achieving the Levenshtein bound, which have parameters $(2^{2m}+2^m,2^m)$ and alphabet $6$. It is interesting to
construct optimal or almost optimal codebooks with different parameters from other tools.

{\bf Acknowledgement.}
This work was supported by SECODE project  and the National Natural Science Foundation of China
(Grant No. 11871058, 11531002, 11701129).  C. Tang also acknowledges support from
14E013, CXTD2014-4 and the Meritocracy Research Funds of China West Normal University. Y. Qi also acknowledges support from Zhejiang provincial Natural Science Foundation of China (LQ17A010008, LQ16A010005).

\end{document}